\documentclass[preprint]{elsarticle}

\journal{arXiv}

\usepackage{lmodern}
\usepackage{url}
\usepackage[bookmarks]{hyperref}
\usepackage{amsmath}
\usepackage{amsfonts}
\usepackage{amsthm}
\usepackage{amssymb}
\usepackage{accents} 
\usepackage[english]{babel}
\usepackage{graphicx}
\usepackage{url}
\usepackage{xspace}
\usepackage[titlenotnumbered,tworuled]{algorithm2e}
\usepackage{graphicx}
\usepackage{xparse}
\usepackage{etoolbox}
\usepackage{xcolor}
\usepackage{mathtools}
\usepackage{colortbl}
\usepackage{epsfig,graphicx}
\usepackage{enumitem}

\newtheorem{theorem}{Theorem}

\newtheorem{definition}[theorem]{Definition}

\usepackage{tikz}
\usetikzlibrary{calc}
\usetikzlibrary{chains}
\usetikzlibrary{patterns}
\usetikzlibrary{decorations.pathreplacing}
\usetikzlibrary{er}
\usetikzlibrary{shapes}
\usetikzlibrary{shapes.multipart}
\usetikzlibrary{arrows}
\tikzset{markovstate/.style={shape=circle,draw=black,align=center,inner sep=0,minimum width=7mm}}
\tikzset{markovstate1/.style={shape=circle,draw=black,align=center,inner sep=0,minimum width=16mm}}
\tikzset{markovstate2/.style={shape=circle,draw=black,align=center,inner sep=0,minimum width=16mm, ultra thick, dotted}}
\tikzset{markovedge/.style={-latex}}
\tikzset{label/.style={}}
\tikzset{pkt/.style={draw,rectangle,fill=gray!15,minimum height=25pt,align=center,font=\footnotesize}}
\tikzset{pktlen/.style={above=-2pt,font=\scriptsize}}

\definecolor{mycolor}{RGB}{0,200,100}

\begin{document}

\begin{frontmatter}

\title{Online Graph Exploration on a Restricted Graph Class: Optimal Solutions for Tadpole Graphs}

\address[label2]{University of Vienna, Faculty of Computer Science, 1090 Vienna, Austria}
\address[label1]{ETH Zurich, 8092 Zurich, Switzerland}

\author[label1]{Sebastian Brandt}
\ead{brandts@ethz.ch}

\author[label2]{Klaus-Tycho Foerster} 
\ead{klaus-tycho.foerster@univie.ac.at}

\author[label1]{Jonathan Maurer}
\ead{maurerjo@ethz.ch}

\author[label1]{Roger Wattenhofer}
\ead{wattenhofer@ethz.ch}

\begin{abstract}
We study the problem of online graph exploration on undirected graphs, where a searcher has to visit every vertex and return to the origin.
Once a new vertex is visited, the searcher learns of all neighboring vertices and the connecting edge weights.
The goal of such an exploration is to minimize its total cost, where each edge traversal incurs a cost of the corresponding edge weight.

We investigate the problem on tadpole graphs (also known as dragons, kites), which consist of a cycle with an attached path.
The construction by Miyazaki et al.\ (The online graph exploration problem on restricted graphs, IEICE Transactions 92-D (9), 2009) can be extended to show that every online algorithm on these graphs must have a competitive ratio of $2 - \varepsilon$, but the authors did not investigate non-unit edge weights.
We show via amortized analysis that a greedy approach yields a matching competitive ratio of $2$ on tadpole graphs, for arbitrary non-negative edge weights.
Moreover, we also briefly discuss the topic of advice complexity on cycle and tadpole graphs.

\end{abstract}

\begin{keyword}
Graph Exploration \sep Online Algorithms
\end{keyword}

\end{frontmatter}

\section{Introduction}
Exploring an unknown graph is considered to be one of the fundamental problem of robotics~\cite{DBLP:conf/icra/BurgardMFST00,DBLP:conf/esa/FleischerT05}: A searcher has to visit all vertices and return to the origin.
As the searcher only has information about the subgraph already explored and the adjacent vertices, the problem is commonly studied from an online perspective.
In other words, the goal is to minimize the competitive ratio of the cost of the actually traversed tour versus a tour of minimum cost.

In this paper, we investigate online graph exploration on \emph{tadpole}~\cite{koh1980graceful} graphs (also known as \emph{dragons}~\cite{truszczynski1984graceful} or \emph{kites}~\cite{kim2006super}, see~\cite{Gallian2017}), which can be visualized as follows: start with a cycle and attach an endpoint of a path to one vertex of the cycle.

\vspace{0.1cm}
\noindent\textbf{Background.}
Rosenkrantz et al.~\cite{greedy} proved that a greedy exploration algorithm has a competitive ratio of $\Theta(\log n)$ on general graphs with $n$ vertices.
Even though their result is over 40 years old, it is not known if algorithms with a competitive ratio of $o(\log n)$ exist---the best known lower bound was $2.5-\varepsilon$~\cite{DBLP:conf/sirocco/DobrevKM12} for some time, recently improved to $10/3-\varepsilon$ for planar graphs~\cite{DBLP:journals/corr/abs-2002-10958}.
However, when all edges have unit traversal cost, a depth-first search (DFS) has a competitive ratio of $\frac{2n-2}{n}<2$, with a matching lower bound of $2-\varepsilon$ for any exploration algorithm~\cite{DBLP:journals/ieicet/MiyazakiMO09,tadpole}.
For the case of $k$ different edge weights, a hierarchical DFS provides a competitive ratio of $2k$, which can be extended to a competitive ratio of $\Theta(\log n)$ for arbitrary edge weights~\cite{DBLP:journals/tcs/MegowMS12}.
Moreover, when considering strongly connected directed general graphs, a greedy algorithm achieves an optimal competitive ratio of $n-1$~\cite{DBLP:journals/tcs/FoersterW16}.
The related problem of searching for a specific vertex was covered by Smula et al.~\cite{DBLP:conf/sirocco/KommKKS15,DBLP:phd/dnb/Smula15}, and there is also work exploring graphs with multiple agents and limited memory respective pebbles, e.g., by Disser et al.~\cite{DBLP:journals/jacm/DisserHK19}.
In the case of multiple agents, the exploration problem on tadpoles is related to searching on three rays, where the agents do not start on the junction~\cite{DBLP:journals/tcs/BrandtFRW20}.

With the general undirected weighted case being an open problem for many decades now, some articles studied graph exploration on restricted graph classes, initiated by the \texttt{ShortCut} algorithm of Kalyanasundaram and Pruhs~\cite{DBLP:journals/tcs/KalyanasundaramP94} for planar graphs.
However, Megow et al.~\cite{DBLP:journals/tcs/MegowMS12} report on a \textit{``precarious issue in the given formal implementation''} of \texttt{ShortCut}, to which end they reformulate and improve \texttt{ShortCut} in their \texttt{Blocking} algorithm, achieving a competitive ratio of $16(2g+1)$ on graphs of genus at most $g$ (e.g., 16 on planar graphs).
Very recently, Fritsch~\cite{fritsch2020online} showed a \texttt{Blocking} parameter modification to be 3-competitive on unicyclic graphs and $5/2 + \sqrt(2) \approx 3.91$-competitive on~cactus~graphs.
Moreover, Megow et al.~\cite{DBLP:journals/tcs/MegowMS12} show in an intricate construction that the \texttt{Blocking} algorithm does \emph{not} have a competitive ratio of $o(\log n)$ in general.
As such, there is no promising candidate algorithm known that could potentially break the $\log n$ barrier on general graphs, but also no insight on the existence of non-constant lower bounds.
Regarding further graph classes, Asahiro et al.~\cite{DBLP:journals/ipl/AsahiroMMY10} showed that nearest neighbor algorithms achieve a competitive ratio of~$1.5$ on cycles, also giving a lower bound of $1.25$.
The latter result was improved by Miyazaki et al.~\cite{DBLP:journals/ieicet/MiyazakiMO09} by utilizing weighted distance computations to reach a competitive ratio of $\frac{1+\sqrt{3}}{2}\approx 1.366$, with a matching lower bound of $\frac{1+\sqrt{3}}{2}-\varepsilon$.

\vspace{0.1cm}
\noindent\textbf{Motivation.}
For graphs with different edge weights, until the very recent result of Fritsch, there were no results between cycles and planar graphs, leaving a rather large gap in the competitive landscape between the ratios of $1.366$ (cycles) and $16$ (planar graphs), where the latter result is also just an upper bound.
We thus investigate a natural extension of the cycle, by joining the endpoint of a path to a cycle, also known as tadpole, dragon, or kite graphs~\cite{Gallian2017}.
Maybe interestingly, this class of graphs was also implicitly used for the lower bound construction on unit edge weight graphs in~\cite{DBLP:journals/ieicet/MiyazakiMO09,tadpole}.
We thus hope that a further investigation of these graphs will lead to a better understanding of graph exploration on special (planar) graphs.

\vspace{0.1cm}
\noindent\textbf{Contribution.}
We prove that a greedy exploration algorithm achieves a competitive ratio of 2 on tadpole graphs, for non-negative edge weights.
We also extend the construction by Miyazaki et al.~\cite{DBLP:journals/ieicet/MiyazakiMO09} to prove a lower bound of $2-\epsilon$ for unit edge weight tadpole graphs, i.e., our results are tight. We also briefly discuss the topic of advice complexity on cycle and tadpole graphs.

\vspace{0.1cm}
\noindent\textbf{Organization.}
Our article is structured as follows. We first provide a formal model in \S\ref{sec:model}, followed by a lower bound proof  in \S\ref{sec:lower} that shows there is no tadpole exploration algorithm with a competitive ratio of $2-\varepsilon$ on unweighted graphs, extending a construction from Miyazaki et al.~\cite[\S~4.2]{DBLP:journals/ieicet/MiyazakiMO09}.
We then provide a proof in \S\ref{sec:algo} that a greedy exploration results in a matching competitive ratio of 2 on tadpoles, discuss the topic of advice complexity in \S\ref{sec:advice}, concluding in \S\ref{sec:conclusion}.

\section{Model}\label{sec:model}
\noindent\textbf{Graph model.}
We consider connected undirected graphs $G=(V,E)$ with $|V|=n$ vertices and $|E|=m$ edges, denoting an edge connecting $u$ and $v$ as $(u,v)$.
Each edge $e \in E$ is equipped with a positive edge weight $c(e)$.\footnote{We can omit edge weights of 0, as these edges can be explored for free in the undirected case, i.e., one might as well ``contract'' these edges and only consider positive weights.} 
In particular, all investigated graphs will be tadpole graphs, which are graphs where one vertex has degree $1$, one vertex has degree $3$, and all other vertices have degree~2.
Hence, one can imagine that a tadpole graph consists of one cycle with $i$ vertices, to which a path (or, from now on, \emph{stem}) of $j$ vertices is attached.
As such, all tadpole graphs can be represented e.g. in the form $T_{i,j}$, with $i\geq 3$ and $j\geq 1$.

\vspace{0.1cm}
\noindent\textbf{Exploration model.}
We assume that the searcher has unlimited computational power and memory.
Each vertex is equipped with an unique identifier (ID), where exploration proceeds as follows~\cite{DBLP:journals/tcs/MegowMS12,DBLP:journals/tcs/KalyanasundaramP94}:
upon arriving at a vertex $v$, the searcher obtains the IDs of all adjacent vertices, as well as the weight of all incident edges.
In order to move to a neighboring vertex $u$ of $v$, the searcher has to pay the edge weight of $(v,u)$, even if this edge was traversed before.
The goal is to perform a closed tour (a closed walk) from the starting vertex $s \in V$ that visits all vertices in $V$, while minimizing the accumulated cost.
We call algorithms that perform such closed tours exploration algorithms.

\vspace{0.1cm}
\noindent\textbf{Competitive ratio.}
The quality of an exploration algorithm $A$ is rated by its competitive ratio.
The competitive ratio $R_{A}(G)$ on a given graph $G$ is defined as $R_{A}(G)=\frac{\text{cost}_A(G)}{\text{cost}_{\text{opt}}(G)}$, where $\text{cost}_A(G)$ is the cost of the tour of the algorithm $A$ on $G$ and $\text{cost}_{\text{opt}}(G)$ is the cost for the optimal tour. 
A competitive ratio of 1 is optimal and given two algorithms, the one with the lower competitive ratio is better. 
One can find an optimal tour cost by solving the traveling salesman problem (TSP)~\cite{travelingsalesman} on the complete graph $G'$ with $V(G')=V(G)$, where for all $u,v \in V(G'), u \neq v$, the edge weight $c((u,v))$ is the length of the shortest path between $u$ and $v$ in $G$~\cite{greedy}.
Lastly, the competitive ratio of $A$ on some graph class $\mathcal G$ is defined as the supremum of its competitive ratio over all graphs $G \in \mathcal G$.

\vspace{-2mm}

\section{Lower Bounds for Tadpole Graphs}\label{sec:lower}
\vspace{-2mm}
To provide some first intuition, we extend the lower bound construction by Miyazaki et al.~\cite[\S~4.2]{DBLP:journals/ieicet/MiyazakiMO09} to the case of tadpole graphs.
We note that Miyazaki et al.~state their theorem for general graphs, but in fact, they only require the explored graph to be either a tree or a tadpole graph.
We now extend their proof such that we can promise to the searcher that the graph to be explored is a tadpole graph, while retaining lower bounds.

\begin{theorem}[Extended construction from~\cite{DBLP:journals/ieicet/MiyazakiMO09}, \S~4.2]\label{thm:lb}
For any positive constant $1 > \varepsilon> 0$, there is no $(2 - \varepsilon)$-competitive online algorithm for unit weight tadpole graphs.
\end{theorem}
We note that our following proof construction heavily relies on~\cite[\S~4.2]{DBLP:journals/ieicet/MiyazakiMO09}. Our adapted main proof idea is as follows:

When starting on the cycle, after the junction vertex $t$ with degree $3$ is found, the graph looks like three paths joined at $t$ to the searcher.
The challenge for the searcher is now to find out which of the two paths form the cycle. 
However, the adversary can ``force'' the searcher to explore the cycle first, i.e., the searcher needs to go back to the earlier visited junction $t$. 
If the adversary appropriately chooses the path lengths according to the searcher's decisions, the additional tour length required can  be arbitrarily close to the optimal tour length.
We now give the promised formal proof, following Miyazaki et al.'s notation:

\begin{proof}
Let the searcher start on a vertex $s$ of degree $2$, where an adversary will construct the tadpole graph on the go.
Until the junction vertex $t$ is visited, the graph will look like two edge-disjoint paths to the searcher, consisting of the vertices $s,v_1,v_2,\ldots$ and $s,u_1,u_2,\ldots$.
Initially, the vertices $v_1,u_1$ are visible to the searcher, but yet unvisited.
The adversary keeps extending the paths, until the searcher reaches a distance of $k$ on one of the paths, w.l.o.g.\ on $u_k$.
We set $u_k$ as the junction $t$ and the adversary reveals two new vertices, adjacent to $t$, namely $q_1$ and $p_1$.
Denote by $v_{t_1}$, $t_1 < k$, the most distant visited vertex on the opposing side of $s$, where the adversary revealed a vertex $v_{t_1+1} \neq t$.
So far, the searcher traversed at least $k+2t_1$ edges.
We show the current situation in Figure~\ref{fig:lb-1}.

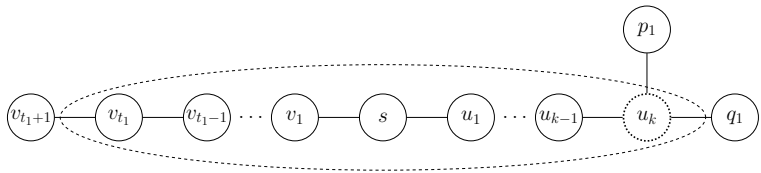
\begin{figure}[h]
\centering
	\scalebox{0.39}{
	\begin{tikzpicture}[auto]
	\node [markovstate1] (s) at (0,0) {\huge $s$};
	\node [markovstate1] (v1) at (-3,0) {\huge $v_1$};
	\node [markovstate1] (u1) at (3,0) {\huge $u_1$};
	\node [markovstate1] (uk-1) at (6,0) {\huge $u_{k-1}$};
	\node [markovstate2] (uk) at (9,0) {\huge $u_{k}$};
	\node [markovstate1] (p1) at (9,3) {\huge $p_{1}$};
	\node [markovstate1] (q1) at (12,0) {\huge $q_{1}$};
	\node [markovstate1] (vt1-1) at (-6,0) {\huge $v_{t_1-1}$};
	\node [markovstate1] (vt1) at (-9,0) {\huge $v_{t_1}$};
	\node [markovstate1] (vt1+1) at (-12,0) {\huge $v_{t_1+1}$};
	\node (dots) at (-4.5,0) {\huge $\ldots$};
	\node (dots) at (4.5,0) {\huge $\ldots$};
	
	\draw [dashed] (0,0) ellipse (11cm and 1.8cm);

	\draw [thick] (s) to (u1);
	\draw [thick] (s) to (v1);
	\draw [thick] (vt1-1) to (vt1);
	\draw [thick] (vt1) to (vt1+1);
	\draw [thick] (uk-1) to (uk);
	\draw [thick] (uk) to (p1);
	\draw [thick] (uk) to (q1);
	\end{tikzpicture}
		}
		\caption{Situation after the searcher arrived on the dotted junction vertex $u_k$, after starting on $s$. The dashed ellipse marks vertices already visited. The three vertices $v_{t_1+1}$, $p_1$, $q_1$ are visible, but not yet visited. Out of the these vertex types, it is not yet clear to the searcher which form the stem and which complete the cycle.}\label{fig:lb-1}
\end{figure}

Next, the searcher has three visible but yet unvisited vertices left: $v_{t_1+1}, p_1, q_1$.
In our construction, either $p_1$ or $q_1$ will lead to the end of the stem, but the searcher does not know which one to follow.
However, the searcher could also approach the problem from ``the other side'' and visit $v_{t_1+1}$ next.
If $p_1$ is visited, then the adversary uncovers another vertex~$p_2$, similarly, if $q_1$ is visited, then the adversary uncovers $q_2$, and if $v_{t_1+1}$ is visited, then the adversary uncovers $v_{t_1+2}$, and so on.
The adversary only needs to reveal if either type $p$ or $q$ vertices form the stem, when every vertex on the cycle is visible to the searcher.
In the following, for ease of argumentation, we will make the searcher more powerful such that once all vertices on the cycle are visible to the searcher, we will reveal the whole graph to the searcher.
%

%
As the by us controlled adversary builds the graph, we now assume that there are $k \geq 4 $ vertices $v_{t_1+1},\ldots,v_{t_1+k}$ between $t$ and $v_{t_1}$ on the cycle, leave the stem length undefined for now, just setting it to be less than $k$, and perform case distinction for the following two different cases.

\bigskip

\begin{enumerate}
	\item \label{case:u1} Until all vertices on the cycle are visible to the searcher, assume that the searcher does not go back through the junction $t$ and the start $s$ to $v_{t_1}$, but just uncovers new vertices of type $p$ and $q$. 
	\begin{itemize}
		\item W.l.o.g., we can assume that this situation occurs when the searcher is on vertex $p_{k-2}$ (where $p_{k-1}$ neighbors $p_k = v_{t_1+1}$, in turn a neighbor to $v_{t_1}$), where the furthest so far visited vertex of type $q$ is $q_{k'}$ with $k-2 > k'$.\footnote{Essentially, as soon as the searcher visits a vertex of distance $k-2$ to $t$, we decide that this vertex path is of type $p$.} So far, the searcher traversed at least $k +2 t_1 + k-2+ 2k'$ edges.
		 We now reveal the whole graph to the searcher and set the length of the stem to be $k'+1$, i.e., the vertex $q_{k'+1}$ is not yet visited.
		The current situation is shown in Figure~\ref{fig:lb-2}
		
		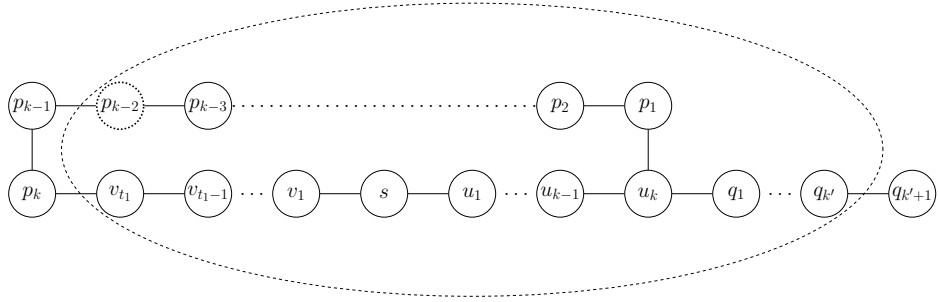
\begin{figure}[h]
\centering
	\scalebox{0.39}{
	\begin{tikzpicture}[auto]
	\node [markovstate1] (s) at (0,0) {\huge $s$};
	\node [markovstate1] (v1) at (-3,0) {\huge $v_1$};
	\node [markovstate1] (u1) at (3,0) {\huge $u_1$};
	\node [markovstate1] (uk-1) at (6,0) {\huge $u_{k-1}$};
	\node [markovstate1] (uk) at (9,0) {\huge $u_{k}$};
	\node [markovstate1] (p1) at (9,3) {\huge $p_{1}$};
	\node [markovstate1] (p2) at (6,3) {\huge $p_{2}$};
	\node [markovstate1] (q1) at (12,0) {\huge $q_{1}$};
	\node [markovstate1] (qk) at (15,0) {\huge $q_{k'}$};
	\node [markovstate1] (qk1) at (18,0) {\huge $q_{k'+1}$};
	\node [markovstate1] (vt1-1) at (-6,0) {\huge $v_{t_1-1}$};
	\node [markovstate1] (vt1) at (-9,0) {\huge $v_{t_1}$};
	\node [markovstate1] (vt1+1) at (-12,0) {\huge $p_k$};
	\node [markovstate1] (vt1+2) at (-12,3) {\huge $p_{k-1}$};
	\node [markovstate2] (vt1+3) at (-9,3) {\huge $p_{k-2}$};
	\node [markovstate1] (vt1+4) at (-6,3) {\huge $p_{k-3}$};
	\node (dots) at (-4.5,0) {\huge $\ldots$};
	\node (dots) at (4.5,0) {\huge $\ldots$};
	\node (dots) at (13.5,0) {\huge $\ldots$};
	\node (dots) at (0,3) {\huge $\ldots \ldots \ldots \ldots \ldots \ldots \ldots \ldots \ldots \ldots \ldots $};
	
	\draw [dashed] (3,1.5) ellipse (14cm and 5cm);

	\draw [thick] (s) to (u1);
	\draw [thick] (s) to (v1);
	\draw [thick] (vt1-1) to (vt1);
	\draw [thick] (vt1) to (vt1+1);
	\draw [thick] (uk-1) to (uk);
	\draw [thick] (uk) to (p1);
	\draw [thick] (uk) to (q1);
	\draw [thick] (vt1+1) to (vt1+2);
	\draw [thick] (vt1+2) to (vt1+3);
	\draw [thick] (vt1+3) to (vt1+4);
	\draw [thick] (qk) to (qk1);
	\draw [thick] (p1) to (p2);
	\end{tikzpicture}
		}
		\caption{Situation after the searcher arrived on the dotted vertex $p_{k-2}$, when the whole graph was revealed. The dashed ellipse marks vertices already visited. The only three unvisited vertices are $p_k = v_{t_1+1}$, $p_{k-1} = v_{t_1+2}$, and the end of the stem, $q_{k'+1}$.}\label{fig:lb-2}
\end{figure}

		Now, it remains to visit the vertices $p_{k-1}=v_{t_1+2}$ and $v_{t_1+1}=p_k$, the stem vertex $q_{k'+1}$, then returning to $s$.
		The shortest path to do so is exactly in this order, i.e., visit  $p_{k-1}=v_{t_1+2}$ and $v_{t_1+1}=p_k$ (2 edges), then return to $t$ ($k$ edges), then visit $q_{k'+1}$ and return to $t$ ($ 2 (k'+1)$ edges), then return to $s$ ($k$ edges).
		In total, the searcher traversed at least $k +2 t_1 + k-2+ 2k' + 2 + k + 2 (k'+1) +k = 4k + 2t_1 + 4k' +2$ edges, whereas the optimal solution traverses every cycle edge once ($2k+t_1+1$) and every stem edge twice ($2(k'+1)$), implying that $\texttt{OPT} = 2k  +t_1 + 2k' +3$.
		The competitive ratio is hence $\geq 2 - 4/(3 + 2k + t_1 + 2k')$. 
		
		As $0 \leq t_1 < k$ and $0 \leq k' < k $, for every fixed $1 > \varepsilon >0$, we achieve a competitive ratio greater than $2 - \varepsilon$ by e.g.\ selecting $k \in \mathbb{N}_{\geq 4}$ such that~$\varepsilon > 4/(3 + 2k)$, i.e., $k > -3/2+2/\varepsilon$.
	\end{itemize}
	
	\item \label{case:u2} It remains to cover the remaining case, i.e., that until all vertices on the cycle are visible to the searcher, the searcher goes at least once back through $t$ via $s$ to $v_{t_1}$.
	If the searcher is then on vertex $p_{k-2}$, we ignore the journey to $v_{t_1}$ (the searcher traversed the edges for free) and refer to Case~\ref{case:u1} above.
	Else,	w.l.o.g., assume that the two remaining unvisited vertices on the cycle are $p_{k_1+1}$ and $p_{k_1+2}$, with $k_1 < k- 2$, where the most distant vertex visited on the stem so far is (analogously as above) w.l.o.g.\ $q_{k_2}$, with $k_2 \leq k_1$.
	We now reveal the whole graph to the searcher, finishing the stem with the unvisited (but visible) vertex $q_{k_2+1}$.
	In other words, it remains to visit $p_{k_1+1}$ and $p_{k_1+2}$, $q_{k_2+1}$, and return to $s$.
	The current situation is shown in Figure~\ref{fig:lb-3}.
		
		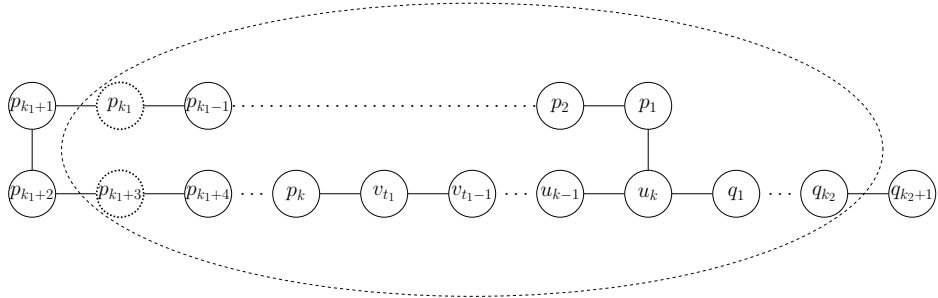
\begin{figure}[h]
\centering
	\scalebox{0.39}{
	\begin{tikzpicture}[auto]
	\node [markovstate1] (s) at (0,0) {\huge $v_{t_1}$};
	\node [markovstate1] (v1) at (-3,0) {\huge $p_k$};
	\node [markovstate1] (u1) at (3,0) {\huge $v_{t_1-1}$};
	\node [markovstate1] (uk-1) at (6,0) {\huge $u_{k-1}$};
	\node [markovstate1] (uk) at (9,0) {\huge $u_{k}$};
	\node [markovstate1] (p1) at (9,3) {\huge $p_{1}$};
	\node [markovstate1] (p2) at (6,3) {\huge $p_{2}$};
	\node [markovstate1] (q1) at (12,0) {\huge $q_{1}$};
	\node [markovstate1] (qk) at (15,0) {\huge $q_{k_2}$};
	\node [markovstate1] (qk1) at (18,0) {\huge $q_{k_2+1}$};
	\node [markovstate1] (vt1-1) at (-6,0) {\huge $p_{k_1+4}$};
	\node [markovstate2] (vt1) at (-9,0) {\huge $p_{k_1+3}$};
	\node [markovstate1] (vt1+1) at (-12,0) {\huge $p_{k_1+2}$};
	\node [markovstate1] (vt1+2) at (-12,3) {\huge $p_{k_1+1}$};
	\node [markovstate2] (vt1+3) at (-9,3) {\huge $p_{k_1}$};
	\node [markovstate1] (vt1+4) at (-6,3) {\huge $p_{k_1-1}$};
	\node (dots) at (-4.5,0) {\huge $\ldots$};
	\node (dots) at (4.5,0) {\huge $\ldots$};
	\node (dots) at (13.5,0) {\huge $\ldots$};
	\node (dots) at (0,3) {\huge $\ldots \ldots \ldots \ldots \ldots \ldots \ldots \ldots \ldots \ldots \ldots $};
	
	\draw [dashed] (3,1.5) ellipse (14cm and 5cm);

	\draw [thick] (s) to (u1);
	\draw [thick] (s) to (v1);
	\draw [thick] (vt1-1) to (vt1);
	\draw [thick] (vt1) to (vt1+1);
	\draw [thick] (uk-1) to (uk);
	\draw [thick] (uk) to (p1);
	\draw [thick] (uk) to (q1);
	\draw [thick] (vt1+1) to (vt1+2);
	\draw [thick] (vt1+2) to (vt1+3);
	\draw [thick] (vt1+3) to (vt1+4);
	\draw [thick] (qk) to (qk1);
	\draw [thick] (p1) to (p2);
	\end{tikzpicture}
		}
		\caption{Situation after the searcher arrived on either the dotted vertex $p_{k_1}$ or $p_{k_1+3}$, previously visiting the other respective vertex, when the whole graph was revealed. The dashed ellipse marks vertices already visited. The only unvisited vertices are $p_{k_1+2}=v_{t_1+k-k_1-1}$, $p_{k_1+1}=v_{t_1+k-k_1}$, and the end of the stem, $q_{k_2+1}$.}\label{fig:lb-3}
\end{figure}

	We now perform a case distinction depending on if the searcher is currently either on the vertex $p_{k_1+3}$ or the vertex $p_{k_1}$.
	\begin{enumerate}
		\item \label{case:u21} The searcher is currently on $p_{k_1+3}$. So far, the searcher traversed at least $k + 2t_1 + 2k_1 + 2k_2 + k + t_1 + 1 + k - k_1-3 = 3k + 3t_1 + 1k_1 + 2  k_2 -2 $ edges. In order to reach the end of the stem on the shortest path (visiting $p_{k_1+1},p_{k_1+2}$ on the way), $k_1+3$ (to the junction) and then $k_2+1$ edges are needed. To return to $s$ on the shortest path, $k_2+1+k$ edges must be traversed.
		Hence, the total route takes at least $3  k + 3  t_1 + k_1 + 2  k_2 -2 + k_1+3+ 2(k_2+1)+k = 4 k + 3  t_1 + 2  k_1 + 4  k_2 +3$ edge traversals.
		The optimal solutions traverses every cycle edge once ($2  k +  t_1+1$) and every stem edge twice ($2(k_2+1)$), setting setting $\texttt{OPT} = 2k  +t_1 + 2k_2 +3$.
		To obtain a lower bound on the competitive ratio, we observe that $4 k + 3  t_1 + 2  k_1 + 4  k_2 +3 > 4  k + 2  t_1 + 4  k_2 +2$, i.e., we obtain again (as in Case~\ref{case:u1}) a competitive ratio of $\geq 2 - 4/(3 + 2 k + t_1 + 2  k')$.

		\item It remains to cover the case that the searcher is currently on $p_{k_1}$.
		In this case, the searcher, after being at the junction $t$, also traveled to $q_{k_2}$ and $p_{k_1+3}=v_{t_1+k-k_1-2}$.
		So far, the searcher traversed at least $k + 2  t_1 + k_1 + 2  k_2 +2(k+t_1+k-k_1-2) = 5  k + 4  t_1 - k_1 + 2  k_2 -4 $ edges. In order to visit the last three unvisited vertices $p_{k_1+2}=v_{t_1+k-k_1-1}$, $p_{k_1+1}=v_{t_1+k-k_1}$, $q_{k_2+1}$ (the end of the stem), and return to $s$, all on a shortest path, the searcher does the following: First visit $p_{k_1+1}=v_{t_1+k-k_1}$, $p_{k_1+2}=v_{t_1+k-k_1-1}$ (2 edges), then the end of the stem ($2+k_1+k_2+1$ edges), then return to $s$ ($k_2+1+k$ edges), i.e., $k + k_1+ 2  k_2 + 6$ additional edges.
		In total, the searcher traversed at least $ 5  k + 4  t_1 - k_1 + 2  k_2 -4 + k + k_1+ 2  k_2 + 6 = 6  k + 4  t_1 + 4  k_2 +2$ edges.
		As before, the optimal solutions traverses every cycle edge once ($2  k +  t_1+1$) and every stem edge twice ($2(k_2+1)$), setting $\texttt{OPT} = 2k  +t_1 + 2k_2 +3$.
		By observing that $ 6  k + 4  t_1 + 4  k_2 +2 \geq 4  k + 2  t_1 + 4  k_2 +2$, we can again repeat the arguments from Case~\ref{case:u1}, selecting $k \in \mathbb{N}_{\geq 4}$ with $k > -3/2+1/\varepsilon$.
	\end{enumerate}
\end{enumerate}
This completes the proof by case distinction.
\end{proof}

On the algorithmic side, a depth-first search achieves a matching competitive ratio for unit weight graphs, but has no guarantees for weighted graphs.
We show 
that a greedy approach yields matching upper bounds for weighted graphs.

\section{Greedy Exploration of Tadpole Graphs}\label{sec:algo}
In this section, we will utilize an amortized analysis to prove that a greedy exploration is at most $2$-competitive on tadpole graphs. We first formalize the notion of a greedy exploration algorithm in \S\ref{subsec:greedy} along with some preliminary definitions and a simple upper cost bound, and then present our proof in \S\ref{subsec:theorem}.

\subsection{Greedy Algorithm Preliminaries}\label{subsec:greedy}

We first define how a greedy algorithm explores the graph, and then introduce the notions of a step and charging edges for a step. 
Afterwards, we show how these definitions yield a simple first upper cost bound.
We will utilize these definitions in the proof of the optimal upper bound statement in \S\ref{subsec:theorem}.

\begin{definition}
A greedy exploration algorithm proceeds as follows, until every vertex is visited: From all vertices that are known but have not yet been visited, pick a vertex $w$ to which the best known path is shortest, and visit $w$ via this shortest path. Once all vertices have been visited, return to the starting vertex $s$ along a shortest path.
\end{definition}
\smallskip
\noindent\textbf{Step by step.} 
It will be useful to introduce the notion of a \emph{step} of the algorithm, where a step is the decision which previously unvisited vertex to visit next, plus the actual visit itself. 
As such, if a graph has $n$ vertices, a greedy online exploration algorithm takes $n-1$ steps and then return to the starting vertex~$s$.
For consistency, we consider the return to $s$ as a step as well, i.e., an exploration algorithm for $n$ vertices will perform $n$ steps in total.

\smallskip
\noindent\textbf{Charging an edge for the next step.}
Consider the situation that the searcher is currently on a vertex $v$ \emph{with} an unvisited neighbor $w$.
If $v$ has more than one unvisited neighbor, let $w$ be the one that connects to $v$ with the cheapest edge.
In the next step, a greedy exploration algorithm currently on a vertex $v$ will either $a)$ visit a yet unvisited neighboring vertex $w$ by traversing $e=(w,v)$ or $b)$ visit a yet unvisited vertex $w'$ by a path $P$ that does not contain $(w,v)$, where not necessarily $w' \neq w$.
Due to the greedy nature of the exploration algorithm, in case $b)$ the cost of $P$ will be at most the cost of $e$, i.e., $c(P) \leq c(e)$ (Fig.~\ref{fig:fig}). 
Similarly, in case $a)$ the cost $c(e)$ is cheaper than any other path to a yet unvisited vertex.
Hence, for an amortized analysis, in such a situation we can say that we \emph{charge} the cost of the next step to the edge $(w,v)$.

\begin{figure}[h]
\centering
	\scalebox{0.65}{
	\begin{tikzpicture}[auto]
	\node [markovstate] (v) at (0,0) {\large $v$};
	\node [markovstate] (w) at (3,0) {\large $w$};
	\node [markovstate] (1) at (-3,0) {};
	\node [markovstate] (2) at (-6,0) {};
	\node [markovstate] (w1) at (-9,0) {\large $w'$};
	\node (dots) at (-4.5,0) {\huge $\ldots$};
	\draw [thick] (v) to node [midway, below, xshift=0mm]{\large$e$} (w);
	\draw [thick] (v) to (1);
	\draw [thick] (2) to (w1);
	\draw [dashed] (-3,0) ellipse (4cm and 0.75cm);
	\draw [thick, bend right = 30, ->, dotted] (v) to node [midway, above]{\large$P$} (w1);
	\end{tikzpicture}
		}
		\caption{Illustration of the case $b)$ when charging an edge for the next step, where the dashed ellipse marks vertices already visited. When the searcher is on $v$ and decides to visit $w'$ instead of $w$, then the cost of the dashed path $P$ is at most the cost of the edge $e=(w,v)$.}\label{fig:fig}
\end{figure}
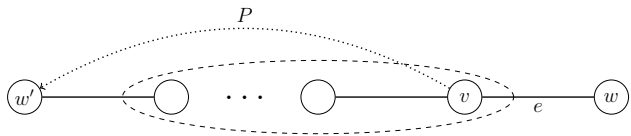

However, an exploration algorithm could also find itself in the situation that after performing a step, all neighboring vertices have been explored. 
In this case, there is no incident edge to charge:

\smallskip
\noindent\textbf{Charging the actual path.}
Consider the situation that the searcher is currently on a vertex $v$ \emph{without} an unvisited neighbor or that all vertices have already been visited, but $v \neq s$.  
Then, for the next step (which can also be the return to $s$), we say that we charge the cost of the actually taken path to said path, i.e., the cost of each traversed edge.
\smallskip
\noindent\textbf{A first upper cost bound.}
We now combine both charging ideas for a first simple upper bound, to give some intuition for the upcoming proof in \S\ref{subsec:theorem}.
To this end, we ask the question how often we cannot charge an edge for the next step---but rather have to charge the actual path taken.
By definition, this can only be the case if there is no unvisited neighbor $w$ of the current vertex $v$.
Such a situation can occur only in two cases:
\begin{enumerate}
	\item The searcher is not on $s$ and at the end of the stem, or
	\item the searcher is on the cycle, and all cycle vertices have been visited.
\end{enumerate}
Regarding the first case, if the end of the stem is $s$, then we will only be on $s$ at the beginning (where we have an unvisited neighbor) or upon completion of the algorithm.

For the second case, assume not all cycle vertices have been visited, but that on the current node $v$ (at the beginning of a step), the searcher has no unvisited neighbors. 
We now show by contradiction that this cannot be the case.
Denote the neighbors of the current vertex $v$ on the cycle by $w'$ and $w''$: when they were visited for the first time, it was either because they were $s$ or because they were at the end of a step. 
Similarly, $v$ was visited for the first time at the end of a step.
However, in order to visit both $w',w''$ without visiting $v$, all vertices between $w',w''$ on the cycle must have been visited, a contradiction.
Hence, both cases will only occur at most once until the algorithm has completed its exploration, as the searcher will not choose an already explored vertex for the endpoint of a step (unless it is the return to $s$ in the last step).
As thus, charging a path will only happen twice, and for the second time, it is the return to $s$, i.e., the final step.
Here it is important to note that the greedy algorithm will only charge an edge $e=(x,y)$ when it connects the current node (w.l.o.g.~$x$) to an unvisited node $y$. 
It can thus never charge an edge $e=(x,y)$ twice\footnote{We note that we defined charging a path differently than charging an edge, even if the path is of length $1$.}, as the searcher will only be on $x$ once at the beginning of a step.
Hence, a greedy exploration algorithm will charge every edge at most once, and charge at most two paths, i.e., the total exploration cost is at most three times all edge weights in the tadpole graph.
We will next improve this simple bound.

\subsection{Achieving 2-competitiveness on Tadpoles via a Greedy Algorithm}\label{subsec:theorem}
Our first simple upper bound at the end of \S\ref{subsec:greedy} is still far from the desired {{2-competitiveness}}: In our amortized analysis, we would like to include the cost of an edge only twice as often as an optimal solution uses it.
Furthermore, an optimal solution might skip a very expensive cycle edge.
To this end, we first investigate the optimal tour shapes and costs.
\smallskip
\noindent\textbf{Shape of the optimal tour.} Observe that the optimal tour for a tadpole graph consists of $a)$ the optimal walk through the cycle~\cite{DBLP:journals/ipl/AsahiroMMY10} and $b)$ twice the attached path (stem).
Thus, an optimal tour has one of the following two shapes:
\begin{enumerate}
	\item Traverse all the edges in the cycle once and all edges in the stem twice.
	\item Traverse all the edges in the cycle, except for one, twice, and the stem twice. The excluded edge $e_\infty$ has a weight of at least the remaining cycle.
\end{enumerate}
To break ties, we say that Shape 1 is optimal, in case both have the same cost. As such, we assume in the second case that $e_\infty$ has a weight \emph{greater} than the remaining cycle.
\smallskip
\noindent\textbf{Cost of the optimal tour.}
For our analysis, denote the $|E|=m$ edges in the tadpole graph by $e_1,\ldots,e_m$, where we set $i)$ the stem edges to be $e_1,\ldots e_j$, $1 \leq j \leq m-3$ and the cycle edges to be $e_{j+1},\ldots,e_m$, and $ii)$ $e_m = e_\infty$ in the case of Shape 2.
It follows that:
\begin{enumerate}
	\item An optimal tour of Shape 1 has a cost of $\sum_{i=1}^j{2  c(e_i)} + \sum_{i=j+1}^m{1  c(e_i)}$,
	\item an optimal tour of Shape 2 has a cost of $\sum_{i=1}^{m-1}{2  c(e_i)}$.
\end{enumerate}

\noindent Using the above optimal tour costs, we now prove the promised upper bound:

\begin{theorem}
A greedy exploration algorithm has a competitive ratio of $2$ on weighted tadpole graphs.
\end{theorem}

\begin{proof}
We prove the theorem statement by case distinction for tadpole graphs of Shape 1 or of Shape 2.
We begin with Shape 2.

\medskip

\noindent\textbf{Shape 2.}
We show the theorem statement for Shape 2 by proving that a greedy exploration algorithm incurs a cost of at most
$  2 \left( \sum_{i=1}^{m-1}{2  c(e_i)} \right)$.
First observe that the greedy exploration algorithm will never traverse the edge $e_\infty$, as its weight is greater than the remaining cycle edges combined.

We now consider the incurred costs of the greedy algorithm until it is, at the end of a step, on a vertex incident to $e_\infty$ or on a vertex with no unvisited neighbors.
Note that this situation will happen up to three times, twice when visiting the endpoints of the edge $e_\infty$ and once when on the end of the stem.
Let this situation occur for the first time after ${t_1}$ steps.
So far, in each step ${h}$, $1 \leq h \leq t_1$, we can charge the cost of the step to some edge $e_{a_h} \in E$, with $e_{a_{h}} \neq e_{a_{h'}}$ for $h \neq h'$.\footnote{We introduce the notation $e_{a_h}$ as the edge $e_{a_h}$ can be different from the edge $e_h$. Recall that we either traverse the edge $e_{a_h}$ directly with cost $c(e_{a_h})$ or traverse a path with cost at most $c(e_{a_h})$, see Figure~\ref{fig:fig}.}
Hence, the cost of step ${h}$ is at most $c(e_{a_h})$, with $e_{a_h} \neq e_\infty$, and after ${t_1}$ steps we incur costs of at most $\sum_{i={1}}^{{t_1}}{c(e_{a_i})}$, with $\{e_{a_1},\ldots,e_{a_{t_1}}\} \subset \{e_1,\ldots,e_{m-1}\}$.
Note that if ${t_1}=0$, the cost so far is 0.
We now consider the cost of step ${t_1}+1$. We can upper bound its cost by $\sum_{i=1}^{m-1}{c(e_i)} $, which suffices for the rest of the analysis.

We next consider the second time when we are on one of the endpoints of~$e_\infty$ or on the end of the stem, after $t_2 \geq t_1+1$ steps.
Using our charging technique, we can bound the cost incurred after step $t_1+1$ to and including step $t_2$ by $\sum_{i={t_1+2}}^{{t_2}}{c(e_{a_i})}$, for $t_2 > t_1+1$, with $\{e_{a_1},\ldots,e_{a_{t_1}}\} \cap \{e_{a_{t_1+2}},\ldots,e_{a_{t_2}}\} = \emptyset$ and $\{e_{a_{t_1+2}},\ldots,e_{a_{t_2}}\} \subset \{e_1,\ldots,e_{m-1}\}$, as we charge every edge at most once.
%
%
If $t_2 = t_1+1$, then this cost is 0.
For step $t_2+1$, we again upper bound its cost being at most $\sum_{i=1}^{m-1}{c(e_i)} $, which suffices for the rest of the analysis.
We now consider the third and last time when we are on one of the endpoints of~$e_\infty$ or on the end of the stem, after $t_3 \geq t_2+1$ steps.
As above, we can bound the cost incurred after step $t_2+1$ to and including step $t_3$ by $\sum_{i={t_2+2}}^{{t_3}}{c(e_{a_i})}$, for $t_3 > t_2+1$, with $\left (\{e_{a_1},\ldots,e_{a_{t_1}}\} \cup \{e_{a_{t_1+2}},\ldots,e_{a_{t_2}}\} \right ) \cap \{e_{a_{t_2+2}},\ldots,e_{a_{t_3}}\} = \emptyset$ and $\{e_{a_{t_2+2}},\ldots,e_{a_{t_3}}\} \subset \{e_1,\ldots,e_{m-1}\}$, as again, we charge every edge at most once.
If $t_3 = t_2+1$, then this cost is 0.
Once both endpoints of $e_\infty$ and the end of the stem are visited (the case after step $t_3$), it remains to return to $s$, i.e., in step $t_3+1$.
We again upper bound its cost to being at most $\sum_{i=1}^{m-1}{c(e_i)} $.
When we sum up all costs, we incur the cost of each edge $e \neq e_\infty$ at most once, plus $3  \left( \sum_{i=1}^{m-1}{c(e_i)} \right )$.
In other words, the greedy algorithm has a cost of at most $4  \left( \sum_{i=1}^{m-1}{c(e_i)} \right ) =  2  \left( \sum_{i=1}^{m-1} 2 {c(e_i)} \right ) $, which completes the analysis for Shape~2.\footnote{Recall that we never charge an edge twice and for Shape 2, we never traverse $e_\infty$.}
It remains to investigate Shape 1.

\medskip

\noindent\textbf{Shape 1.}
We show the theorem statement for Shape 1 by proving that a greedy algorithm incurs a cost of at most
$  2 \left( \sum_{i=1}^j{2  c(e_i)} + \sum_{i=j+1}^m{1  c(e_i)} \right)$, where the edges $e_{1},\ldots,e_{j}$ form the stem of the tadpole graph.
%
%
%
We will analyze Shape~1 by case distinction, depending which of the following situations occur:
\begin{enumerate}[leftmargin=3.5em]
	\item[(1)] The starting vertex $s$ is not the end of the stem.
	\item[(2)] The starting vertex $s$ is the end of the stem.
\end{enumerate}

\noindent Case~(1) is divided into two subcases:

\begin{enumerate}[leftmargin=3.5em]
	\item[(1a)] The searcher visits the end of the stem before visiting all vertices on the~cycle.
	\item[(1b)] The searcher visits all vertices on the cycle before visiting the end of the~stem.
\end{enumerate}

The underlying idea for this case distinction is as follows: until the searcher visits the end of the stem or has visited all vertices in the cycle, there is always an unvisited neighbor $w$ such that from the current vertex $v$, we can charge the so far uncharged edge $e=(v,w)$ for the cost of the next step with cost at most~$c(e)$.

\begin{itemize}
	\item 
	We start with Case (1a), where the searcher visits the end of the stem (not being $s$), before visiting all vertices on the cycle.
	\quad In this case, assume the searcher performed $t_1$ steps so far. 
	In each step $1 \leq h \leq t_1$, the searcher can charge the cost of the step to an incident edge $e_{a_h}$, where $e_{a_h} \neq e_{a_h'}$ for $h \neq h'$, i.e., with a cost of at most $c(e_{a_h})$.
	\quad Next, consider the step $t_2 > t_1$ where the searcher visits for the first time after step $t_1$ a yet unvisited vertex $v_{t_2}$ on the cycle, possibly the junction $t$.
	The cost of the steps $t_1+1,\ldots,t_2$ is exactly the cost of traversing all stem edges once $(c(e_{1})+\ldots+c(e_{j}))$ plus the cost of traversing the shortest path between $t$ and $v_{t_2}$, being zero cost for $t = v_{t_2}$ and else consisting of the pairwise distinct cycle edges $e_{b_1},\ldots,e_{b_k}$.
	\quad We now consider the steps $t_2+1$ to $t_3$, where after step $t_3$ the searcher has no incident edge to charge, i.e., all vertices on the cycle are visited. Using the charging scheme, the cost of the steps can be bounded by the summed up cost of $t_3-t_2$ further pairwise distinct cycle edges $e_{a_{t_2+1}},\ldots,e_{a_{t_3}}$, with $\{e_{a_{t_2+1}},\ldots,e_{a_{t_3}}\} \cap \{ e_{a_{1}},\ldots,e_{t_1} \} = \emptyset$, as here we just charge edges we have not charged in the steps $1, \ldots, t_1$.
	\quad It remains to return to $s$ in step $t_3+1$. Let the summed up cost of all cycle edges be $C$. Observe that we can bound the cost of step $t_3+1$ by the summed up cost of all stem edges $e_{1},\ldots,e_j$ plus at most $0.5  C$. The latter holds as the shortest path between any two vertices on the cycle has always a cost of at most $0.5  C$. Analogously, it holds that $0.5  C \geq \sum_{i=1}^{k}{c(e_{b_i})}$.
	\quad We now consider the total cost incurred in this case, where it holds that $\sum_{i=1}^{t_1}{c(e_{a_i})}+\sum_{i=t_2+1}^{t_3}{c(e_{a_i})} \leq \sum_{i=1}^{m}{c(e_i)}$, as we charge each tadpole edge at most once in the steps $1,\ldots,t_1,t_2+1,\ldots,t_3$.
	For the steps $t_1+1$ to~$t_2$, the total cost is $\sum_{i=1}^{j}{c(e_i)}+\sum_{i=1}^{k}{c(e_{b_i})}$, with $\sum_{i=1}^{k}{c(e_{b_i})} \leq 0.5  C$. With the final step $t_{3}+1$ incurring a cost of at most $0.5  C+ \sum_{i=1}^{j}{c(e_i)}$ it holds that:
	$\sum_{i=1}^{m}{c(e_i)} + \sum_{i=1}^{j}{c(e_i)} + 0.5  C + 0.5  C+ \sum_{i=1}^{j}{c(e_i)} 
	= 
	\sum_{i=1}^{m}{c(e_i)} +2 \left( \sum_{i=1}^{j}{c(e_i)} \right) + \sum_{i=j+1}^{m}{c(e_i)}
	\leq 2 \left( \sum_{i=1}^j{2  c(e_i)} + \sum_{i=j+1}^m{1  c(e_i)} \right)$~.

	\item We now cover the Case (1b), where we do not start at the end of the stem, and visit all vertices on the cycle before visiting the end of the stem.
	Assume that the searcher so far performed $t_1$ steps, where we charge $t_1$ pairwise distinct edges for a total cost of at most $\sum_{i=1}^{m}{c(e_i)}$. 
	\quad Next, assume that the searcher reached the end of the stem by step $t_2$. Observe that in the steps $t_1+1,\ldots,t_2$, the searcher simply walks to the end of the stem on a simple shortest path, as all cycle vertices are already explored. The incurred costs are at most $0.5  C + \sum_{i=1}^{j}{c(e_i)}$. 
	\quad It remains to return to the starting vertex $s$. Again, we can bound the cost of doing so by at most $0.5  C + \sum_{i=1}^{j}{c(e_i)}$.
\quad The total incurred costs are at most $\sum_{i=1}^{m}{c(e_i)} + 2 \left(  \sum_{i=1}^{j}{c(e_i)} \right) + C = 
%
3 \left( \sum_{i=1}^{j}{c(e_i)} \right) +  2 \left( \sum_{i=j+1}^{m}{c(e_i)} \right)
\leq 2 \left( \sum_{i=1}^j{2  c(e_i)} + \sum_{i=j+1}^m{1  c(e_i)} \right)$~.

	\item It remains to cover the Case (2) where the searcher starts on the end of the stem. In this case, it will take $n-1$ steps until the last vertex (on the cycle) is visited. For each step, a different edge will be charged, leading to a cost bounded by $\sum_{i=1}^{m}{c(e_i)}$. Afterwards, the searcher returns to $s$ on a shortest path, where we can again bound the cost by $\sum_{i=1}^{m}{c(e_i)}$.
\end{itemize}


As such, we showed 2-competitiveness for all cases in both Shape~1 and Shape~2.
\end{proof}

\section{Beyond 2-Competitiveness with Advice}\label{sec:advice}
In order to obtain an improved competitive ratio, one can leverage the model of \emph{advice complexity} for graph exploration~\cite{DBLP:conf/waoa/BockenhauerFU18,DBLP:journals/corr/abs-1804-06675}.
Herein, an all-knowing oracle can equip the searcher with a bit string before initial execution (the advice), which the exploration algorithm can leverage for improved performance.

\smallskip
\noindent\textbf{2 bits of advice for tadpole graphs.}
Already 2 bits of advice suffice to obtain a competitive ratio of $\frac{1+\sqrt{3}}{2}\approx 1.366$, where, if the searcher does not start on the junction, the first bit indicates if the searcher starts on the stem  or not.
In case the searcher does not start on the stem and not on the junction, then the second bit can be used to indicate which neighboring vertex of the junction forms the beginning of the stem, as the vertices have unique identifiers.
To this end, as the oracle knows the deterministic algorithm of the searcher, it knows which two choices remain at the junction, where a choice from a local ordering of the 2 identifiers can be performed with 1 bit.
Lastly, if the searcher starts on the junction, then the 2 bits can be used to indicate which choice leads to the stem.
Equipped with these 2 bits, the searcher acts as follows. 
When starting on the stem (not on the junction), one proceeds in an arbitrary direction until either the end of the stem is reached, in which case one turns around towards the junction, or until the junction is reached. In that case, the $\frac{1+\sqrt{3}}{2}$ competitive cycle exploration algorithm by Miyazaki et al.~\cite{DBLP:journals/ieicet/MiyazakiMO09} is used (where we use the junction vertex $t$ as the starting point), after which the searcher proceeds on a shortest path to $s$, possibly after visiting the end of the stem and returning to $s$. By doing so, the searcher traverses every edge of the stem twice (which is optimal), and explores the cycle itself with a competitive ratio of $\frac{1+\sqrt{3}}{2}$.
Hence, the total competitive ratio is at most $\frac{1+\sqrt{3}}{2}$.
When not starting on the stem and not on the junction, the searcher runs the $\frac{1+\sqrt{3}}{2}$-competitive algorithm of Miyazaki et al.~\cite{DBLP:journals/ieicet/MiyazakiMO09} as well, but when reaching the junction for the first time, the searcher explores the stem by traversing every stem edge twice (using the second advice bit), and then continues with the cycle exploration. Again, the competitive ratio is at most $\frac{1+\sqrt{3}}{2}$.
When starting on the junction, the searcher uses the 2 bits to explore the stem first and return to $s$ (each stem edge twice) and then uses the $\frac{1+\sqrt{3}}{2}$-competitive algorithm of Miyazaki et al.~\cite{DBLP:journals/ieicet/MiyazakiMO09} for the cycle exploration, finishing the case distinction. We cast our insights into the following theorem:

\begin{theorem}
2 bits of advice suffice to explore weighted tadpole graphs with a competitive ratio of $\frac{1+\sqrt{3}}{2}$.
\end{theorem}

\smallskip

Before giving an advice scheme that provides optimal exploration for tadpole graphs, we first present how to obtain optimal exploration for cycles.

\smallskip

\noindent\textbf{Optimality for cycles.} 
%
%
In order to obtain $1$-competitiveness for cycles, it suffices to mark the edge~$e_\infty$ for Shape 2 (traverse all edges except $e_\infty$ twice) respectively to indicate if the current graph is of Shape 1 (traverse all edges in the cycle exactly once). 
As the oracle knows all actions of the deterministic searcher ahead of time, it knows in which order the $m$ edges will be made visible to the searcher. 
In each step, only one edge will be revealed, except at the start, where we assume that the edge incident to the neighbor with lower identifier is the first in this order.
We now show how $\lceil\log_2(n)\rceil$ bits suffice to mark the edge and also indicate Shape 1, on cycle graphs, where $n=m$ holds.
In order to indicate Shape 1, we simply mark an edge incident to $s$ with lowest weight, which cannot be $e_\infty$, as $e_\infty$ has a greater weight than the remaining cycle, and the searcher traverses every edge once, being optimal for Shape 1.
Else, the searcher knows that Shape 2 was indicated, and traverses the cycle in a deterministic direction (e.g., starting with the neighbor with a lower identifier), until the edge $e_\infty$ is revealed. Then, the searcher turns around, explores the cycle until the edge $e_\infty$ appears again, and returns to the start, being optimal for Shape 2.
Hence the following theorem holds:

\begin{theorem}
$\lceil \log_2(n)\rceil$ bits of advice suffice to explore weighted cycle graphs with a competitive ratio of $1$.
\end{theorem}
The usage of 2 additional bits, as before, could now be leveraged to obtain a competitive ratio of 1 on tadpole graphs, i.e., with $\lceil \log_2(n)\rceil + 2 $ bits.
However, we can improve this result to only require 1 extra bit.

\medskip

\noindent\textbf{Optimality for tadpoles.} 
We now assume that Shape 1 is indicated if all bits from the first $\lceil \log_2(n)\rceil$ bits of advice are set to zero, and that else Shape 2 is indicated to the searcher.\footnote{Unlike in cycles, in tadpole graphs, the starting node $s$ could be on the stem and hence marking the cheaper incident edge for Shape 1 does not suffice, as both incident edges could have the same weight.}
We now show how $\lceil \log_2(n)\rceil+1$ bits of advice suffice for optimality on tadpole graphs.

\smallskip

In the case of Shape 2, the searcher can virtually remove $e_{\infty}$ from the graph, as for the cycle, and the remaining graph essentially consists of three stems. Note that due to using $\lceil \log_2(n)\rceil+1$ bits, we have sufficient bits to indicate which revealed edge is $e_{\infty}$, even while reserving advice with the first $\lceil \log_2(n)\rceil$ bits as zero for Shape 1.
The searcher can detect all three ends of the stems (one is of degree 1 and the other two are endpoints of $e_{\infty}$). 
When the searcher reaches each end of the stem, the searcher simply backtracks. 
At the junction, the searcher may encounter two unknown edges, but either one may be traversed first, so no additional advice is needed.

\smallskip

In the case of Shape 1, the first $\lceil \log_2(n)\rceil$ bits of advice are set to zero (there is no edge to avoid traversing) and we need to use one additional bit of advice at the junction, to indicate in which direction (of the two) the searcher should proceed next. 
The searcher's strategy is as follows:
\begin{itemize}
	\item At the beginning, the searcher departs in any direction.
	\item When the searcher is at the junction for the first time after the beginning, the searcher follows the one bit of advice.
	\item When the searcher reaches the end of the stem, the searcher backtracks.
\end{itemize}
In case the starting vertex is on the stem (except for the junction), the extra bit of advice may be arbitrary (or may be ignored), because both of the new two edges at the junction are on the cycle.
In case the starting vertex is on the cycle (except for the junction), the advice bit indicates the edge on the stem and the searcher traverses the stem first. 
Finally, suppose that the starting vertex is the junction. If the edge chosen at the beginning is on the stem, then the searcher will know this fact when he reaches the end of the stem.
If the chosen edge is on the cycle, then he will know this fact when he comes back to the junction, so he may go towards the third (unknown) edge. In either case, the additional one bit may be ignored.
We again cast our insights into the following theorem:

\begin{theorem}
$\lceil \log_2(n) \rceil +1 $ bits of advice suffice to explore weighted tadpole graphs with a competitive ratio of $1$.
\end{theorem}

\section{Conclusion and Outlook}\label{sec:conclusion}
We studied the online exploration of tadpole graphs, showing that a greedy exploration achieves a competitive ratio of 2.
Our results also hold on weighted graphs, with a matching lower bound of $2-\varepsilon$ already on unit weight tadpoles.
The latter extends the results of Miyazaki et al.'s~\cite[\S~4.2]{DBLP:journals/ieicet/MiyazakiMO09} lower bound, which required the graph to be either a tree or a tadpole.

Moreover, we have also presented the first non-trivial graph class where a greedy graph exploration is optimal from a competitive point of view, and provided first insights into the advice complexity of cycle and tadpole graphs.

We see our work as a step towards charting the landscape of graph exploration, a problem which continues 
 to puzzle researchers for many decades.
Next research directions could be the investigation of more broad graph classes, such as lollipop graphs, barbell graphs, or wheel graphs, or the bound improvements on unicyclic and cactus graphs, with the hope that a better understanding of these specific graph structures will bring us closer to bridging the large gap for exploration on general graphs.

\vspace{0.5cm}
\noindent\textbf{Bibliographical note.}\\
A preliminary version of the results in this article appeared in~\cite{Maurer2015}.

\vspace{0.5cm}
\noindent\textbf{Acknowledgements.}\\
We would like to thank Nicole Megow for her feedback on the related work.
We would also like to thank anonymous reviewers for their helpful comments, which among other details, lead to a better presentation of the upper bound proof. Moreover, it was also pointed out to us that Miyazaki et al.'s~\cite[\S~4.2]{DBLP:journals/ieicet/MiyazakiMO09} lower bound proof required that the graph be either a tadpole or a tree, and we could hence strengthen their result such that the lower bound construction only requires tadpole graphs.
Lastly, an anonymous reviewer also provided us with a proof on how the advice complexity for tadpole graphs could be lowered from $\lceil \log_2(n)\rceil+2$ to $\lceil \log_2(n)\rceil+1$ bits, for which we would like to express our thanks as well.

\bibliographystyle{elsarticle-num}
\bibliography{references}

\begin{thebibliography}{10}
\expandafter\ifx\csname url\endcsname\relax
  \def\url#1{\texttt{#1}}\fi
\expandafter\ifx\csname urlprefix\endcsname\relax\def\urlprefix{URL }\fi
\expandafter\ifx\csname href\endcsname\relax
  \def\href#1#2{#2} \def\path#1{#1}\fi

\bibitem{DBLP:conf/icra/BurgardMFST00}
W.~Burgard, M.~Moors, D.~Fox, R.~G. Simmons, S.~Thrun, Collaborative
  multi-robot exploration, in: {ICRA}, {IEEE}, 2000, pp. 476--481 (2000).

\bibitem{DBLP:conf/esa/FleischerT05}
R.~Fleischer, G.~Trippen, Exploring an unknown graph efficiently, in: {ESA},
  Vol. 3669 of Lecture Notes in Computer Science, Springer, 2005, pp. 11--22
  (2005).

\bibitem{koh1980graceful}
K.~Koh, D.~Rogers, H.~Teo, K.~Yap, Graceful graphs: some further results and
  problems, Congr. Numer 29 (1980) 559--571 (1980).

\bibitem{truszczynski1984graceful}
M.~Truszczynski, Graceful unicyclic graphs, Demonstatio Mathematica 17 (1984)
  377--387 (1984).

\bibitem{kim2006super}
S.~Kim, J.~Y. Park, On super edge-magic graphs, Ars Combinatoria 81 (2006)
  113--127 (2006).

\bibitem{Gallian2017}
J.~A. Gallian,
  \href{http://www.combinatorics.org/ojs/index.php/eljc/article/viewFile/DS6/pdf}{A
  dynamic survey of graph labeling (twenty-first edition, 2018)}, The
  Electronic Journal of Combinatorics (2018).
\newline\urlprefix\url{http://www.combinatorics.org/ojs/index.php/eljc/article/viewFile/DS6/pdf}

\bibitem{greedy}
D.~J. Rosenkrantz, R.~E. Stearns, P.~M. Lewis, II, An analysis of several
  heuristics for the traveling salesman problem, SIAM journal on computing
  6~(3) (1977) 563--581 (1977).

\bibitem{DBLP:conf/sirocco/DobrevKM12}
S.~Dobrev, R.~Kr{\'{a}}lovic, E.~Markou, Online graph exploration with advice,
  in: {SIROCCO}, Vol. 7355 of Lecture Notes in Computer Science, Springer,
  2012, pp. 267--278 (2012).

\bibitem{DBLP:journals/corr/abs-2002-10958}
A.~Birx, Y.~Disser, A.~V. Hopp, C.~Karousatou, Improved lower bound for
  competitive graph exploration, CoRR abs/2002.10958 (2020).

\bibitem{DBLP:journals/ieicet/MiyazakiMO09}
S.~Miyazaki, N.~Morimoto, Y.~Okabe, The online graph exploration problem on
  restricted graphs, {IEICE} Transactions 92-D~(9) (2009) 1620--1627 (2009).

\bibitem{tadpole}
N.~Morimoto, {{Design and Analysis of Algorithms for Graph Exploration and
  Resource Allocation Problems and Their Application to Energy Management
  (Kyoto University)}}, Ph.D. thesis (2014).

\bibitem{DBLP:journals/tcs/MegowMS12}
N.~Megow, K.~Mehlhorn, P.~Schweitzer, Online graph exploration: New results on
  old and new algorithms, Theor. Comput. Sci. 463 (2012) 62--72 (2012).

\bibitem{DBLP:journals/tcs/FoersterW16}
K.-T. Foerster, R.~Wattenhofer, Lower and upper competitive bounds for online
  directed graph exploration, Theor. Comput. Sci. 655 (2016) 15--29 (2016).

\bibitem{DBLP:conf/sirocco/KommKKS15}
D.~Komm, R.~Kr{\'{a}}lovic, R.~Kr{\'{a}}lovic, J.~Smula, Treasure hunt with
  advice, in: {SIROCCO}, Vol. 9439 of Lecture Notes in Computer Science,
  Springer, 2015, pp. 328--341 (2015).

\bibitem{DBLP:phd/dnb/Smula15}
J.~Smula, Information content of online problems: advice versus determinism and
  randomization, Ph.D. thesis, {ETH} Zurich, Switzerland (2015).

\bibitem{DBLP:journals/jacm/DisserHK19}
Y.~Disser, J.~Hackfeld, M.~Klimm, Tight bounds for undirected graph exploration
  with pebbles and multiple agents, J. {ACM} 66~(6) (2019) 40:1--40:41 (2019).

\bibitem{DBLP:journals/tcs/BrandtFRW20}
S.~Brandt, K.-T. Foerster, B.~Richner, R.~Wattenhofer, Wireless evacuation on
  \emph{m} rays with \emph{k} searchers, Theor. Comput. Sci. 811 (2020) 56--69
  (2020).

\bibitem{DBLP:journals/tcs/KalyanasundaramP94}
B.~Kalyanasundaram, K.~Pruhs, Constructing competitive tours from local
  information, Theor. Comput. Sci. 130~(1) (1994) 125--138 (1994).

\bibitem{fritsch2020online}
R.~Fritsch, Online graph exploration on trees, unicyclic graphs and cactus
  graphs, CoRR abs/2004.06690 (2020).

\bibitem{DBLP:journals/ipl/AsahiroMMY10}
Y.~Asahiro, E.~Miyano, S.~Miyazaki, T.~Yoshimuta, Weighted nearest neighbor
  algorithms for the graph exploration problem on cycles, Inf. Process. Lett.
  110~(3) (2010) 93--98 (2010).

\bibitem{travelingsalesman}
D.~L. Applegate, R.~E. Bixby, V.~Chvatal, W.~J. Cook, The traveling salesman
  problem: a computational study, Princeton University Press, 2011 (2011).

\bibitem{DBLP:conf/waoa/BockenhauerFU18}
H.~B{\"{o}}ckenhauer, J.~Fuchs, W.~Unger, Exploring sparse graphs with advice
  (extended abstract), in: {WAOA}, Vol. 11312 of Lecture Notes in Computer
  Science, Springer, 2018, pp. 102--117 (2018).

\bibitem{DBLP:journals/corr/abs-1804-06675}
H.~B{\"{o}}ckenhauer, J.~Fuchs, W.~Unger, The graph exploration problem with
  advice, CoRR abs/1804.06675 (2018).

\bibitem{Maurer2015}
J.~Maurer, {{Graph Exploration, thesis, ETH Zurich, Switzerland}} (2015).

\end{thebibliography}

\end{document}